\documentclass{article}
\usepackage{amsmath, amssymb}
    \usepackage{url}

\usepackage{fullpage}

\newcommand{\ignore}[1]{}
\newtheorem{claim}{Claim}
\newtheorem{theorem}{Theorem}

\newtheorem{lemma}{Lemma}

\newtheorem{definition}{Definition}

\newtheorem{corollary}{Corollary}

\DeclareMathOperator{\winp}{win}
\DeclareMathOperator{\winn}{{\tt winn}}

\DeclareMathOperator{\argmax}{argmax}
\DeclareMathOperator{\vm}{v}
\DeclareMathOperator{\win}{\tt win}
\DeclareMathOperator{\supp}{{\tt sup}}

\DeclareMathOperator{\nwin}{{\cal WIN}}
\DeclareMathOperator{\np}{{\cal P}}
\DeclareMathOperator{\taker}{taker}

\DeclareMathOperator{\pp}{p}
\DeclareMathOperator{\p}{p}
\DeclareMathOperator{\pn}{pn}
\DeclareMathOperator{\vp}{v}

\newenvironment{proof}{\par\noindent{\bf Proof:}}{\hfill$\Box$\\}

\title{Optimistic-Conservative Bidding in Sequential Auctions}

\author{
Avinatan Hassidim\thanks{Bar-Ilan University and Google. {\tt avinatanh@gmail.com}}
\and Yishay Mansour\thanks{Blavatnik School of Computer Science, Tel
Aviv University. This research was supported in part by The Israeli
Centers of Research Excellence (I-CORE) program, (Center  No. 4/11),
by a grant from the Israel Science Foundation, by a grant from
United
  States-Israel Binational Science Foundation (BSF), and by a grant
  from the Israeli Ministry of Science (MoS). {\tt mansour@cs.tau.ac.il}} }

\begin{document}
\maketitle

\begin{abstract}
In this work we consider selling items using a sequential first
price auction mechanism. We generalize the assumption of conservative bidding to extensive form games (henceforth optimistic conservative bidding), and show that for both linear and unit demand
valuations, the only pure subgame perfect equilibrium where buyers are
bidding in an optimistic conservative manner is the minimal Walrasian
equilibrium.

In addition, we show examples where without the requirement of
conservative bidding, subgame perfect equilibria can admit a variety of unlikely predictions, including high price of
anarchy and low revenue in markets composed of additive bidders, equilibria which elicit all the surplus as revenue, and more. We also show that the order in which the items are sold can influence the outcome.

%Getting the Walrasian revenue requires choosing the order in which the items are sold correctly, and we show that the wrong choice of order can have dramatic consequences for the revenue.
\end{abstract}
\thispagestyle{empty}
\newpage
\setcounter{page}{1}

\section{Introduction}

In everyday economy almost all items are sold at a price set by the
seller, given her belief on the demand. In some case, especially
when the seller does not have any credible information regarding the
demand, items are auctioned between multiple potential buyers.
For example, rare art artifacts are often auctioned in auction houses, such as
Sotheby's or Christie's.
While in theory the auction houses could have conducted a huge
combinatorial auction, this is rarely the road taken.
Items are simply auctioned individually at some order (decided by
the auctioneer). Even when there are multiple items sold at once
(e.g., in US spectrum auctions, or off-shore drilling rights), all
the participants know that in the future there will be more auctions
for the good, and strategies accordingly.\footnote{There are
numerous other examples of items which are auctioned sequentially,
such as the Japanese Fish auctions, the Dutch flower auctions, and
Internet ad auctions, such as the ones performed by Yahoo's Right
Media and Google's Ad-Exchange. Bidders may have combinatorial
valuations for the objects at sale, such as an advertiser who wants
a campaign in which a user doesn't see the same ad twice, or a
restaurant which needs to buy enough fish for the day.}

We consider selling
multiple items sequentially by holding a separate first price auction with no reserve price for each
item.
We stress that our goal is not to propose new mechanisms and analyze
their benefits, but rather concentrate on the simple existing per
item auctions and analyze their weaknesses and strengths.
%  between the

Our solution concept is a sub-game perfect equilibrium (SPE), where
each buyer has a strategy that depends on the history. We assume a
full information model, where the seller and each buyer knows
everyone else's valuation. Note that although the seller knows the
buyers' valuations, she is limited in her ability to influence the
outcome since her only action is to select the order in which the
items are  sold.\footnote{The reason we do not allow for reserve
price is this full information model. While clearly the more
interesting case is the Bayesian model, solving the full information
model is a necessary step.}

The starting point of this work is those of Leme et al. \cite{LemeST12-itcs,LemeST12-soda}, who showed the following results:

\begin{enumerate}
  \item Existence: Every sequential auction has a pure equilibrium.
  \item Sort of good news: If all buyers are unit demand, then the price of anarchy of pure equilibria is $2$.
  \item Bad news: If all buyers are submodular, then the price of anarchy is $\Omega(m)$, where $m$ is the number of items.
\end{enumerate}

Our first contribution is to show a necessary and sufficient
condition for an allocation and prices to be a pure SPE for the
sequential auction. This condition is very weak, and we use it to
show that there are many equilibria that do not make sense, e.g., an
equilibrium for submodular bidders in which the social welfare is
half of the optimal welfare, but the seller takes it all as revenue.

Worse, we show an example with two additive bidders in which the
price of anarchy is $\Theta(m)$, where $m$ is the number of
items.\footnote{We also show that the PoA is at most $\Theta(m)$ for
subadditive bidders, but we are more worried about the lower bound
for additive bidders here.} Finally, we show that we can not hope to
have any non trivial revenue guaranteers on the auction, even when
there are two additive bidders.

Motivated by the bad examples of the additive bidders, we generalize
the notion of conservative bidding
\cite{BhawalkarR11,christodoulou2010bayesian} to sequential
auctions. The challenge here is twofold. First, conservative bidding
is a natural notion when players are additive, and requires
adjustment for non-additive bidders. Fortunately, there are several
previous works that extended this notion, see e.g.
\cite{feldman2012simultaneous,fu2012conditional}. The bigger
challenge is that in a sequential auction part of the reason a buyer
may want to buy an item is that it leads to a more desired branch of
the game tree, and we do not want to rule that out. Therefore, we
define {\em optimistic conservative} bidding, which allows a bidder
to place a bid such that if he miraculously wins this bid, he is no
worse than he is on the equilibrium path. This notion coincides with
conservative bidding for a single item. See Appendix \ref{app:con}
for a discussion on why we chose this solution concept, and its
relation to trembling hand equilibrium.

The notion of optimistic conservative bidding already allows us to
get that the unique equilibrium when all buyers are additive is the
``correct'' one, where the buyer with the highest valuation for each
item gets it and pays the second price.

Unfortunately, already for unit demand buyers optimistic
conservative bidding is not enough to guarantee the optimal outcome,
and indeed in the bad example of Leme et
al.\cite{LemeST12-soda} the players bid in an
optimistic conservative way. It turns out that the order of selling
the items has a crucial effect on both the revenue and the welfare.
This immediately raises the algorithmic challenge, of choosing the
order:

\noindent {\em Given the valuations of the players, and the outcomes
of previous auctions,\footnote{Note that we allow the seller to look
at the outcome of the previous auction before choosing what to sell
next. This is not a problem in the solution concept - imagine that
the seller first commits to a function which determines which item
is sold next and then an optimistic conservative SPE is chosen} what
is the optimal way to choose which item to sell next?}

We solve this challenge  for unit demand bidders, showing that if
the selling order is chosen correctly, the only pure optimistic
conservative equilibrium is the optimal outcome, and (perhaps more
important), the minimal Walrasian revenue (equal to the VCG revenue
in this case). In Appendix \ref{app:order} we discuss the selling
order, and give an example where it can have a factor of $m$ on the
revenue, for unit demand bidders.

\subsection{Related Work}

The equivalence between Walrasian pricing and parallel first price
auctions is presented in Bikhchandani \cite{Bikhchandanil99},
showing that there is a one-to-one correspondence between Walresian
equilibrium and pure Nash equilibria of selling the items using
parallel first price auctions.

The work of Hasidim et al. \cite{HassidimKMN11} discusses parallel
first price auctions and Walresian pricing. The main motivation
there was to understand the case of mixed equilibria in the parallel
first price auctions, in the cases where there is no Walresian
equilibrium. This work was later improved by \cite{feldman2012simultaneous}.

the study of sequential auctions in the CS literature was initiated by Leme et al. \cite{LemeST12-itcs}, which we already discussed. Feldman et al.~\cite{feldman2013limits} improved the lower bound on the price of anarchy, and showed that it is at least $\min(n,m)$ in a market where the bidders are either unit demand or additive. Syrgkanis and Tardos~\cite{syrgkanis2012bayesian} extended the model to the Bayesian case. In later work~\cite{syrgkanis2013composable} they proposed using composeable mechanisms to study a series of auctions performed together (in parallel or sequentially).

Sequential auctions were also studied extensively in the economic literature\cite{boutilier1999sequential,hausch1986multi}, where they discuss
different strategies that auction houses use to order the goods for sale \cite{ashenfelter1989auctions,bernhardt1994note}, and the change of price
between identical items sold one after the other \cite{gale1994bottom,mcafee1993declining}. The results there are not algorithmic, but are somewhat aligned with the selling order we present.

\subsection{Paper Organization}

The paper is organized  as follows. We start with some preliminaries
and notation. Then in Section \ref{sec:additive} we show the bad SPE
for additive players, and also that there is a single optimistic
conservative equilibrium. After this motivating example, we show the
characterization of the space of all possible equilibria (and
mention some bad examples deferring proofs to the appendix).
Finally, we move the more involved part of the paper, showing the
optimistic conservative equilibrium for unit demand players, and
proving that it is the unique pure equilibrium.

Following the advice in  the call for papers, in addition to a
motivating example we have included discussions in the potential
weak spots of the paper as well. One may not like the extension of
conservative bidding to extensive form games, and we point to
Appendix \ref{app:con} for a discussion. Another question revolves
around extending the result to more complicated valuations. One of
the challenges here is demand reduction, which is one of the pain
points in auction theory \cite{milgrom2004putting}. See Appendix
\ref{app:complex} for a discussion on why demand reduction is
relevant here, and what is needed to analyze it. Finally, one may be
worried that the seller has full information. We are less worried
about this, since the seller can only choose the next item to sell,
and studying the effect of the order on the revenue is an important
question even if the seller is ignorant.

\section{Model and Preliminaries}

\noindent{\bf Basic setting:} We assume that there is a single
seller with a set $M$ of $m$ items to sell. There is a set $N$ of
$n$ buyers who would like to buy the items. The items are sold sequentially
in first price auctions. We allow the seller to choose which item to sell next,
and how to set the tie breaking rule, as a function of the allocation of the previous items sold.\footnote{If one is content to settle for
$\epsilon$ equilibrium for arbitrarily small $\epsilon$, the tie breaking rule does not matter.} For an item $j$, we will usually denote $\win(j)$ the buyer who won item $j$.

Each buyer $i\in N$ has
a valuation $v_i(S)$ for every set $S\subset M$.
We assume that $v_{i}(\emptyset)=0$, for normalization.
Given a set of items $S\subset M$ for price $p(S)$ the utility of
buyer $i$ is quasi-linear, $u_i(S)=v_{i}(S)-p(S)$. Also, we assume
that buyers are risk neutral, namely, given a distribution over
outcomes, their utility is their expected utility.

A buyer is {\em additive} (also called {\em linear}), if there are
non-negative values $v_{i,1}, \ldots , v_{i,m}$ such that the
valuation is $v_i(S)=\sum_{j\in S} v_{i,j}$ and $v_i(\emptyset)=0$.

A buyer $i$ has a {\em unit demand} valuation function,
%(we also say that buyer $i$ is unit demand),
if there are non-negative values $v_{i,1}, \ldots , v_{i,m}$ such that the valuation is
$v_i(S)=\max_{j\in S} v_{i,j}$ and $v_i(\emptyset)=0$.
If a unit demand buyer already received a set of items $S$, we say that his value for an additional item $f$ is
the marginal value of that item $v_i(S \cup \{f\}) - v_i(S)$, and denote this by
$v_{i,f|S}$.

%Buyer $i$ is {\em single minded} if there is a set $T_i$ and value
%$\lambda_i>0$ such that $v_i(S)= \lambda_i$ if $S\supset T_i$ and otherwise
%$v_i(S)=0$.
%
%Buyer $i$ is {\em $k$-size single minded} if it is single minded and
%$|T_i|=k$.

\smallskip
\noindent{\bf Solution concepts.} We assume a full information
model, where both the seller and each buyer know everyone's
valuation functions. Note that although the seller knows everything, the
influence she has is limited
%thing she can do with this information
to choosing the selling order.

We consider subgame perfect equilibrium. A {\em subgame perfect
equilibrium} (SPE) specifies a strategy for each buyer, such that
the set of strategies are in equilibrium from any possible state
(i.e., history), including states which are not reachable by the
combined set of strategies.\footnote{One may be suspect that since
the seller can choose the selling order dynamically, then the seller
should also be a part of the SPE. This is not the case - first the
seller commits to a function, which determines what item to sell
next given the previous auctions, and then the buyers play some
SPE.} It was shown by \cite{LemeST12-soda} that every sequential
auction has a pure SPE. For the most part, we consider in this work
strategies that depend on the history only through the allocation of
items (and not, for example, on specific bids).

{\bf Walrasian Equilibrium} was proposed by Leon Walras as early as 1874 as a solution concept, and is define as follows:

\begin{definition}
A {\em Walrasian Equilibrium} is a sets of prices $p_j$ for
$j\in[1,m]$ and allocation $\win(j)\in[1,n]$ where $\win(j)$ is the
winner of item $j$, such that:
(1) Each buyer $i$ does a best response, i.e., $v_{i,k}-p_k \geq
v_{i,j}-p_j$,where $win(i)=k$ or if $i$ does not receive any items
then $v_{i,j}\leq p_j$  for any item $j$.
(2) For each item $j$ we have $win(j)=i$ for some buyer $i$.
We define the sum of the prices $\sum_j p_j$ as the Walrasian
revenue.
\end{definition}

\noindent {\bf Conservative bidding} When characterizing auctions,
one of the main issues is our assumptions regarding the buyers that
would submit `loosing bids'.
A common way to handle  this issue in a one-shot auction is to
assume that the buyers bid {\em conservatively}, namely, they always
bid below their valuation. This notion does not immediately fit
sequential auctions. For example, consider an auction with two unit
demand buyers $A,B$ and two items $1,2$, where the selling order is
$1$ and then $2$. If buyer $A$ is not interested in $2$, and buyer
$B$ prefers $2$ to $1$, then a non zero bid of buyer $B$ on $1$ is
not a conservative move - he can anticipate that he will get item
$2$, and does not benefit to win $1$ as well.

%When considering a sequential auction, the natural extension is to
%consider the marginal valuation and define a conservative bidding to
%be at most the marginal valuation. Namely, assuming that we have a
%subset $S$ of items already allocated, then buyer $i$ would not bid
%for item $j$ more than $\vp_{i,j|S}$. Definitely such a bid
%guarantees that winning the current auction would not result in
%negative utility.

%However, there is another aspect of a sequential equilibrium. A
%buyer can foresee what `should' be his utility if the realization
%would result in the on-equilibrium path (from the current state). In
%such a case, the buyer, who bids conservatively using $\vp_{i,j|S}$,
%is taking the risk of losing future utility.

We define the {\em optimistic conservative} as a guarantee to the
buyer that winning with his current bid would not result is a lower
utility than that on the equilibrium path. This is `optimistic'
since the buyer is assigning high likelihood that the equilibrium
path would be realized. Formally, assuming that we have a subset $S$
of items already allocated, then buyer $i$ would not bid for item
$j$ more than $\vp_{i,j|S}-u_{i|S}$,where $u_{i|S}$. (Note that the
allocation $S$ specifies a node in the game tree, since we assume
that the nodes of the game depend only on the outcomes of the
auctions, i.e., who wins.)
% {\bf is the utility assuming the path?}.
An optimistic conservative bidding guarantees that a buyer would
never be worse-off winning an item (which he is not suppose to
win).\footnote{ For a discussion on optimistic conservative bidding,
and the relation to the trembling hand equilibrium, we refer the
reader to the appendix.}

{\bf Price of Anarchy} is the ratio of the maximum social welfare to the worse equilibrium allocation,
where the valuation of an allocation is the sum of buyers valuation.
%(the allocation that maximizes )

\section{Additive Valuations}\label{sec:additive}

It is often the  case that when all buyers have linear (additive)
valuation functions, one can look at each item independently, and
the analysis becomes simple, inheriting many of the nice properties
of single item auctions. Surprisingly, this is not the case for
sequential auctions.

\begin{theorem}
\label{thm:linear} When all buyers have additive valuation
functions, the Price of Anarchy of a sequential auction is
$\Theta(m)$, where $m$ is the number of items. The lower bound on
the PoA holds even when there are two buyers and all the items are identical (so the selling
order is irrelevant).
\end{theorem}

The above theorem would follow from Theorem~\ref{thm:PoA-subadd} and
Lemma \ref{lem:PoA-add}. It is interesting to compare
Theorem~\ref{thm:linear} with the results of \cite{LemeST12-soda},
which showed a PoA of 2 for unit demand bidders, and showed that the
price of anarchy is $\Omega(m)$ for sub-modular buyers.

The upper bound on the PoA in Thorem \ref{thm:linear} actually holds
for sub-additive buyers, as the following theorem shows:

\begin{theorem}
\label{thm:PoA-subadd}
When all buyers have sub-additive valuations, the Price of Anarchy of a sequential auction is $O(m)$, where $m$ is the number of items.
\end{theorem}

\begin{proof}
Let $v_{i,j}$ be the valuation buyer $i$ assigns to item $j$, if
this is the only item that she gets. Let $M = \max_{i,j} v_{i,j}$,
and let $i^*, j^*$ be such that $v_{i^*,j^*} = M$. Since the buyers
are sub-additive, the social welfare is bounded by $m \cdot M$.

Consider the equilibrium path. If the valuation of buyer $i$ bundle
is more than $M/2$, we are done. Else, it must be that the price of
$j^*$ is at least $M/2$ (since buyer $i^*$ prefers the equilibrium
path to snatching $j^*$ when it comes up for sell). In this case,
the revenue is lower bounded by $M/2$, and so is the social
welfare.
\end{proof}

\begin{lemma}
\label{lem:PoA-add}
The Price of Anarchy for additive buyers is $\Omega(m)$ even for identical items and for any selling order.
\end{lemma}

\begin{proof}
Consider a market with two additive buyers $A,B$, and $m$ items,
denoted $1, \ldots, m$. We have that for any $i$
\[V_{A,i} = m\;\;\; \textrm{ but }\;\;\; V_{B,i} = 1\]
As the items are identical, the selling order is irrelevant, and we therefore assume that they are always sold in the order $1$ to $m$.

It is enough to present a bad equilibrium. While buyer $B$ received all items that were on sale so far, both buyers bid $1$
%(equivalently  $b_A = b_B = 1$)
and ties are broken in favor of $B$ on the first $m - 1$ items and
in favor of $A$ in the last item. If $A$ has at least one item both
buyers bid $m$, and ties are broken in favor of $A$.

In the described strategy $A$ has a utility of $m-1$ while $B$ has a
utility of $0$. Observe that if $A$ outbids $B$ on any item, he wins
all the remaining items and has a utility of zero for the additional
items. This implies that $A$ will end with a utility strictly less
than $m-1$. Clearly $B$ has no incentive to deviate since he will
always have utility $0$.\footnote{The fact that this is a SPE also
follows from our characterization of subgame perfect equilibria
(Theorem \ref{folks})} It is easy to see that it obtains a social
welfare of $2m - 1$ where the maximal social welfare is $m^2$.
\end{proof}

In the market we presented, at least the revenue behaves correctly,
in the sense that the revenue is the Walrasian revenue. The
following example shows that this is not always the case.

\begin{theorem}
There exists a sub-game perfect equilibrium for additive buyers where the ratio of its revenue to the Walresian revenue is vanishing.
\end{theorem}

\begin{proof}
As before, consider a market with two additive buyers $A,B$, and $m$
items, denoted $1, \ldots, m$. We have that $V_{A,i}=m$ for any item
$i$ and $V_{B,i}=1$ for $i\leq m-1$ and $V_{B,m} = \epsilon$. Given
that item $m$ is always sold last, the following is an equilibrium:

While buyer $B$ receives all items that are on sale so far, both
buyers bid $\epsilon$ (i.e.,  $b_A = b_B = \epsilon$) and ties are
broken in favor or $B$ on the first $m - 1$ items and in favor of
$A$ in the last item. If $A$ has at least one item both buyers bid
$m$, and ties are broken in favor of $A$.

In this equilibrium, the revenue is $\epsilon m$, which is
arbitrarily lower than $m - 1 + \epsilon$, which is the Walrasian
revenue.
\end{proof}

In contrast to the negative results in the general case, we show
that when the buyers are bidding optimistic conservative, we get a
PoA of $1$, and the revenue the minimal Walrasian revenue. Let the
{\tt max tie-breaking} rule break a tie between a subset of buyers
that bid the same value for a given item $i$ as selecting
$\arg\max_{j\in S} v_{i,j}$. Note that the max tie-breaking has a
different priority for each item.

\begin{theorem}\label{uniq-additive}
For the max tie breaking rule, for any selling order the only pure
optimistic conservative equilibrium is the one in which each item
$i$ is allocated to a buyer in $\arg \max_j v_{j,i}$ and the price
is the second highest valuation.
\end{theorem}

\begin{proof}
The proof is by induction.

\noindent{\bf Base case}
Let $f$ denote the last item being sold. Regardless of the
allocation of the previous items, we have a first price auction for
a single item. Since in the last item optimistic conservative
bidding coincides with conservative bidding, no buyer $j$ can bid
more than $v_{j,f}$. Let $w$ be the buyer that maximizes $v_{j,f}$,
and let $s$ be the buyer with the second highest valuation, that is
\[ s = \argmax_{j \neq w }v_{j,f}\]
Since buyers are conservative, buyer $j$ can not bid over $v_{j,f}$. Therefore, the equilibrium price is at most $v_{s,f}$.

If the equilibrium price is $p < v_{s,f}$, then at least one of $s$
or $w$ would be better-off bidding $\frac{p + v_{s,f}}{2}$ than
their current bid. Therefore, we have $p = v_{s,f}$.

If the price is $v_{s,f}$, then either

\begin{enumerate}
\item Buyer $w$ values $f$ strictly more than $s$, that is $v_{w,f} > v_{s,f}$ and $w$ always buys
\item There are several (more than one buyer) with value $v_{w,f}$ for the item, and one of them always buys
\end{enumerate}
If we assume no ties, then the first option is the only equilibrium.
If ties are possible, then there are several equilibria, but they
are all equivalent (item $f$ has the same price and goes to a buyer
in $\arg\max_j v_{j,f}$ who desires it the most).

\noindent{\bf Inductive step} In the inductive step we assume that
no matter who gets item $f$, for any subsequent item $g$, the buyer
who gets $g$ is one of the buyers who want it the most, i.e.,
$\arg\max_j v_{j,g}$, and would pay the second valuation of $g$.

In this case, the buyer which is allocated $f$ is independent of the
rest of the auction, and an argument very similar to the base case
shows that one of the buyers in $\arg\max_j v_{j,f}$ will be
allocated $f$ and pay the second highest valuation.
\end{proof}

{\bf Remark:} The uniqueness also holds for mixed equilibria, via an
inductive argument. The base case appears in Appendix \ref{app:first}. For the inductive step we note that no matter who wins the item the auction
stays the same, and therefore we invoke the base case again (just like the inductive step in the proof of Theorem \ref{uniq-additive} applies the base case).
%The base case is almost identical to the base case of the proof of Theorem \ref{uniq-unit}, and the inductive step is similar to the base case, as was in the proof of Theorem \ref{uniq-additive}.

\section{Equilibrium Characterization}\label{sec:characterization}

To motivate the definition of optimistic conservative bidding, we show the large variety of possible pure subgame perfect equilibria in sequential first price auctions. To do this, we provide a characterization of this space.

Let $\win : M \rightarrow N$ be an allocation function, and let
$p(1), \ldots p(m)$ be prices. It would be convenient to define the
set of items buyer $i$ wins out of the first $j$ items sold using
the permutation $\pi$ as $x_i(\pi,j)=\{k \;: k\in\{\pi_1,
\ldots,\pi_j\} , \win(k) = i\}$. Let
\[u_i(\win,p) = v_i (x_i(\pi,m)) - \sum_{k\in x_i(\pi,m)} p(k)= \sum_{k\in x_i(\pi,m)} v_{i,k}-p(k)\]
be the utility of buyer $i$ at the end of the auction assuming
everyone played according to $\win$ and $p$. One possible deviation for
buyer $i$ is to play according to the equilibrium up to item $j$,
 to snatch item $j$ by bidding $p(j) + \epsilon$, and to bid zero on all subsequent items. The following theorem shows that there
is an SPE with allocation $\win$ and prices $p$ if and only
if this deviation is not profitable for some permutation $\pi$ of the items. Formally,
%after the first $j$ items (remember that
%$w$ is the allocation and $p$ the prices).

\begin{theorem}\label{folks}
Fix an
allocation $\win$ and prices $p$ such that for any buyer $i$ we have
$u_i(\win,p)\geq 0$. Suppose that there exists a permutation $\pi$ (which describes the selling order assuming everyone plays the equilibrium path) such that for every buyer $i$ and for every
item $j$
\begin{eqnarray}
u_i(\win,p) \ge v_i \left(x_i(\pi,j-1)\cup \{\pi_j\}\right) - p(\pi_j) -  \sum_{k\in x_i(\pi,j-1))}
p(k)\;. \label{eq:general-cond}
\end{eqnarray}
Then there is an SPE with allocation $\win$ and prices $p$.

On the other hand, if for every permutation $\pi$ there exists a buyer $i$ and an item $j$ such
that Equation (\ref{eq:general-cond}) does not hold, then there is
no pure SPE with $\win$ and $p$.
\end{theorem}

\begin{proof}
We present the strategies for the SPE assuming
(\ref{eq:general-cond}) holds where $\pi$ is the identity. Items are
always sold in the order $1, \ldots, m$ regardless of who bought
previous items. On the equilibrium path, at step $j$ (that is when
auctioning item $j$), all buyers bid $p(j)$ and ties are broken in
favor of buyer $\win(j)$. Off the equilibrium path, at some state
$G$ of the game let $X_i(G)$ denote the set of items buyer $i$ has
at state $G$. Given that at state $G$ we are selling item $j$, we
define the marginal valuation of buyer $i$ at state $G$ as,
\begin{equation} \label{marginal}
\vm_{i,j|G} = v_i(\{X_i(G)\} \cup \{j\}) - v_i(\{X_i(G)\})\;.
\end{equation}
In state $G$ each buyer $k$  bids $b_k(G)=\max_i \vm_{i,j|G}$, and ties
are broken in favor the buyer which maximizes (\ref{marginal}).
We prove an inductive claim on the utilities of the buyers in the
off equilibrium part of the game tree:

\begin{claim}
\label{claim:G} Consider an off-equilibrium state $G$ in which item
$j$ is sold. If the buyers play according to the SPE rooted at $G$
no buyer will gain positive utility, and all deviations give either
zero or negative utility.\footnote{Note that the buyers may have gained positive utility before
reaching $G$, we only claim that from $G$ onwards no utility will be
gained.}
\end{claim}

\begin{proof}
The proof is by induction from the leaves to the root. Consider a
leaf state $G$ where item $m$ is sold. Let $k$ be the buyer which
maximizes $\vm_{k,m|G}$. In state $G$ we have a first price auction
for item $m$, where buyer $k$ wins the item and pays $\vm_{k,m|G}$,
which gives her zero utility. If buyer $k$ bids less, she does not
get the item (as everyone else bids $\vm_{k,m|G}$) and therefore has
zero utility, and if she bids more she has negative utility. For any
buyer $r \neq k$, underbidding does not affect the outcome, and
overbidding buyer $k$ would give negative marginal utility.
\end{proof}
%an undesired item.

%The inductive part follows a similar argument.
Now we need to prove the inductive step. By induction hypothesis all
the children vertices of state $G$ give zero utility to all the
buyers. Similar to the leaf case, that there is no local action
which gives a positive utility in state $G$ to any buyer.

If Condition (\ref{eq:general-cond}) does not hold for any
permutation, then there is no SPE, since at least one of the buyers
would deviate.
\end{proof}

We present two surprising corollaries of this theorem. We view this weird equilibria as more evidence that one must refine the set of possible equilibria in a sequential first price auction, and not as plausible outcomes.

Define {\em  non-singleton } valuation, where $v_i(\{j\})=0$ for
any $j\in M$. Note that this includes as a special case, single
minded buyers with sets of size at least two.

\begin{corollary} \label{non-single}
Assume that we have $n$ buyers who have non-singleton valuations
then there is a sub-game perfect equilibrium, where all
the prices are zero and one buyer receives all items.
\end{corollary}

The proof of this corollary is simple, and appears in the full version.

In contrast to the previous example, where the prices were very low, one can build equilibria in which the prices are high,
 and the auctioneer gets a lot of revenue. Consider a market with sub-modular buyers. Let
$\pi$ be the following selling order, defined inductively. Let $i_1,
j_1$ the buyer and item that maximize $\max_{i,j} v_i(\{j\})$. Now at
step $k$, let $i_k, j_k$ be the pair of buyer and item who maximize
\[\vm_{i,j|G} = v_i(\{X_i(G)\} \cup \{j\}) - v_i(\{X_i(G)\} )\]
where $G$ is the state of the game where buyer $i_r$ wins item $j_r$ for every
$r < k$.

\begin{corollary}\label{high-revenue-submodular}
Assume all buyers have a submodular valuation.
If the items are sold according to $\pi$, there is an SPE which elicits at least half
the optimal social welfare as revenue. Moreover, in this SPE the players have no utility, and all the welfare goes to the auctioneer.
%{\bf [[YM: full
%revenue is another word for SW, right?! What is OPT ?! Is OPT also
%SW?! - added a definition to the blurb in the section]]}
\end{corollary}
This corollary is proven in Appendix \ref{app-opt}.
%These funny equilibria (as well as many others that can be constructed), motivate us to consider limiting the space of equilibria, to conservative optimistic ones. Before seeing how this limitation affects the market, we recall the Walrasian equilibrium, and then show when we can predict it.
\section{Unit Demand valuations}

\subsection{Walrasian Equilibrium for Unit Demand Buyers}\label{sec:Walras}

Walrasian equilibrium exists when all buyers are gross substitute. We focus in this section on Walrasian equilibrium when all the buyers are unit demand, to enable us to state the order in which we sell the items. We begin with the definition of supporting a price. Consider a Walrasian equilibrium
defined by the allocation $\win$ and a price vector $p_1, \ldots p_m$.
\begin{definition}
Buyer $i$ supports the price of item $j$ if $\win(j)\neq i$ and either: (1)
there is an item $k$ such that $\win(k)=i$ and $v_{i,k}-p_k =
v_{i,j}-p_j$, (2) otherwise (buyer $i$ gets no item) $v_{i,j}=
p_j$.
\end{definition}
Note that the definition of support almost implies conservativeness.
Next we define a minimal Walrasian Equilibrium.
\begin{definition}
A minimum Walrasian Equilibrium is a set of price $p_1, \ldots ,
p_m$ such that for any other Walrasian equilibrium $p_1', \ldots ,
p_m'$, we have $p_i\leq p_i'$ for any item $i$.
\end{definition}

We show that there exists a minimum Walrasian equilibrium when the buyers
are unit demand.\footnote{This is true in general but not necessary for this work.}

\begin{theorem}
For unit demand buyers there always exist minimum Walrasian
prices.
\end{theorem}

\begin{proof}
It is well known that Walrasian prices maximize the social welfare.
For simplicity, assume there is a unique allocation $\win$ that maximizes the social welfare.
Assume that $p^1_1, \ldots , p^1_m$ and $p^2_1, \ldots , p^2_m$ are
two Walrasian prices for the allocation $\win$. Define $p$ as
$p_i=\min\{p^1_i,p^2_i\}$. We show that $p$ is also a Walrasian
equilibrium.

Assume some buyer $i$ would like to deviate from $\win$ when the
prices are $p$. Let $k$ be the item allocated to buyer $i$, i.e.,
$\win(k)=i$ and assume that $p_k=p^1_k$. After the deviation buyer
$i$ prefers item $j$. If the price of item $j$ in $p$ and $p^1$ are
equal, buyer $i$ can deviate in $p^1$, contradicting that $p^1$ is
equilibrium prices. If the price of item $j$ in $p$ and $p^2$ are
equal, then the utility of buyer $i$ from item $k$ in $p$ is at
least as large as in $p^2$, and hence it would have benefited at
least as much in the deviation in $p^2$, contradicting that $p^2$ is
equilibrium prices. Therefore $p$ and $\win$ are a Walrasian
equilibrium. Therefore, there exists a minimal vector of prices.
\end{proof}

\begin{claim}
\label{claim:Wals_min}
A minimum Walrasian equilibrium has for each
item $j$ with a strictly positive price ($p_j>0$) a buyer $i$ that
supports it.
\end{claim}

\begin{proof}
Assume that there is an item $j$ with $p_j>0$ and no supporter. Let
$u_i$ be the utility of buyer $i$ in the equilibrium. Let
$\epsilon=\min_{i\neq win(j)} u_i-(v_{i,j}-p_j)$. Since item $j$
does not have any supporter $\epsilon>0$. Consider an $\epsilon/2$
decrease in the price of item $j$. Buyer $\win(j)$ definitely still
prefers item $j$. Any other buyer $i$, since it was not a supporter
of item $j$, then after the $\epsilon/2$ decrease it still does not
prefer item $j$. Therefore we showed that the original Walrasian
equilibrium is not minimum.
\end{proof}

We can now define the order of the support in a minimum Walrasian
equilibrium.

\begin{definition}
A {\em support order} of a Walrasian equilibrium is a permutation
$\pi$ of the items, such that for any item $j$ the buyer $i$ that
supports it receives a later item or no item.
\end{definition}

\begin{theorem}
A minimum Walrasian equilibrium has a support order $\pi$.
\end{theorem}

\begin{proof}
We build the permutation $\pi$ from the end to the start. The last
items are items that are priced at zero (if there are any) in an
arbitrary order. All the remaining items have strictly positive
prices, and therefore have at least one supporter (by
Claim~\ref{claim:Wals_min}).

Let $D$ be the set of buyers that either buy at price zero or do not
buy and $I$ be the set of items they buy. If there is an item $i$
such that $\win(i)\not\in D$ and $support(i)\in D$, we add item $i$
to $I$ and to the permutation $\pi$, and add $\win(i)$ to $D$. The
process can terminate either if we exhaust all the items, in which
case we have a permutation $\pi$, or the set of remaining items do
not have supporters in $D$, i.e., for any item  $i\not\in I$ and any
supporter we have $support(i)\not\in D$.
We will show that such an event will contradict the fact that we
have a Walrasian equilibrium.

In such a case for any $j\in D$ we have that $v_{j,k}-p_k >
v_{j,i}-p_i$ for the item $k$ buyer $j$ buys ($\win(k)=j$) and any
item $i\not\in I$. Let
$$\epsilon=\min_{j\in D,i\not\in I, win(k)=j } (v_{j,k}-p_k) -
(v_{j,i}-p_i) $$
We are guarantee that $\epsilon>0$. Consider
reducing the prices of all items not in $I$ by $\epsilon/2$. The
preference between items not in $I$ by buyers not in $D$ does not
change. No buyer in $D$ will prefer an item not in $I$, by definition
of $\epsilon$. A buyer not in $D$ will not prefer an item in $I$
since he has an item he prefers not in $I$ and therefore after
lowering the price his preference only strengthen.

This results in an Walrasian equilibrium with lower prices, which
contradict that we started with a minimum Walrasian equilibrium.
\end{proof}

\subsection{Walrasian sub-game prefect equilibrium ({\tt Unit-Wlrs-Eq})}

In this section we describe the strategy {\tt Unit-Wlrs-Eq}. We
basically define the tree strategy by defining for each sub-tree the
on-equilibrium strategy. This defines the strategy in each node. (We
need to define for each node which item is sold and what are the
bids.)

Given the set of all items $I$ we compute the minimum Walrasian
equilibrium and its support order. The on-equilibrium strategy is
the support order of the items. Let $u_i$ be the utility of buyer
$i$ in the Walrasian equilibrium, using the marginal utilities. The
bid of buyer $i$ on item $j$ on the equilibrium path is
$v_{i,j}-u_i$, if $i$ did not win an item yet and $0$ otherwise.
Note the winning buyer and the supporting buyer both bid the
Walrasian equilibrium price.

We now complete the off-equilibrium path strategy. If a non-winner
buyer bids below his expected bid
or the winner buyer bids above his expected bid
we stay on the equilibrium path (the allocation did not change). If either the winner bids
below the price (losing the item, which we call {\em underbidding}) or another bidder bids higher than
the price (winning the item, which we call {\em overbidding}) we basically construct recursively a
strategy for the remaining items.

Before we complete the off-equilibrium strategy, we first define the
residual valuation of a buyer. Given the already allocated items $S$
we define for each buyer a marginal valuation,
\[\vp_{i,j|S} = \max(0,v_{i,j} - \max_{j\prime \in S, \ win(j\prime) = i} v_{i, j\prime})\]

Give the set of items already sold, we compute the minimum Walrasian
equilibrium $W_{I-S}$ on $I-S$ with valuation $\vp_{i,j|S}$, where
$\win^{I-S}$ is the allocation $p^{I-S}$ are the prices.
Let $u_{i|S}$ be the utility of agent $i$ in $W_{I-S}$ with respect
to  $\vp_{i,j|S}$. (This is essentially the marginal utility from
the item he receives in $W_{I-S}$.)
Give $W_{I-S}$ we compute the support order of $W_{I-S}$ and sell
items in that order. The bid of buyer $i$ which has already received
his item in $W_{I-S}$ is zero. The bid of buyer $i$ which has not
received yet his item in $W_{I-S}$ is $v_{i,j|S}-u_{i|S}$.
Again, note that the winner and supporter of item $j$ bid the new
Walrasian price, i.e., $p^{I-S}_j$.

We continue this recursive process until we define the entire
strategy tree. In Section \ref{sec:proof} We prove that this is indeed an equilibrium

\begin{theorem} \label{thm:equilibrium}
The resulting strategy is sub-game perfect.
\end{theorem}

The equilibrium strategies are optimistic conservative by construction.

\subsection{Proof of Theorem \ref{thm:equilibrium}}\label{sec:proof}

The proof is by induction on the height of the node in the strategy
tree from the leaf (which is also equal to the number of items
left). The following is the inductive claim.

%The following is the inductive hypothesis.
\begin{claim}
\label{claim:SGP:ind}
Every subtree of height $\ell$ (from the
leaves) defines a subgame perfect equilibrium.
\end{claim}

we start by proving the base of the inductive claim
(Claim~\ref{claim:SGP:ind}).

\begin{claim}
The leaves of the strategy tree (subtrees of height $1$) define an
equilibrium strategy.
\end{claim}

\begin{proof}
In a leaf node $r$ where an item $j$ is sold.
%we have a single item $i$ that is auctioned.
The minimal Walrasian equilibrium is the
second highest valuation $v_{i,j|S}$ where $S=I- \{j\}$. Since there
is no continuation, each buyer $i$, except for the highest valuation
buyer, bids $v_{i,j|S}$. Let $i_r$ be the highest valuation buyer
and he bids $b_{i_r}=\max_{i\neq i_r}\{v_{i,j|S}\}$.

This is an equilibrium, since any buyer $i\neq i_r$ bidding higher
than $b_{i_r}$ would result in negative utility. Since this is a
leaf, a subgame perfect equilibrium coincide with equilibrium.
\end{proof}

In the inductive step, some subset of items $S$ has been sold. the current item on sale is $f$. We show two lemmas

\begin{lemma}\label{no-underbid}
Assume that Claim~\ref{claim:SGP:ind} holds for any subtree of
height at most $\ell-1$. Then in any node of height $\ell$
underbidding is weakly dominated.
\end{lemma}

Lemma \ref{no-underbid} follows from the following lemma

\begin{lemma}\label{lem:f_prefer}
  Buyer $\win(f)$ weakly prefers getting $f$ for $p_f$ to any minimal Walrasian equilibrium generated after $f$ is given to any other buyer.
\end{lemma}

If there are no ties, then this is a strong preference.

\begin{proof}
Denote by $\taker$ the buyer who received the item, and let the new Walrasian equilibrium be denoted as $\nwin$ and $\np$.

If buyer $\win(f)$ does not buy any item at $\nwin$, or buys an
item $g_1$ such that $\np_{g_1} \ge p_{g_1}$
then $\win(f)$ does not strictly prefer the new equilibrium. Therefore, we only
need to consider the case that for some item $g_1$ we have
$\nwin(g_1) = \win(f)$ and $\np(g_1) = p_(g_1) -\delta$ for some
$\delta > 0$.

To reach a contradiction, we assume that there is a maximum size set
$D$ of $k\geq 1$ items, such that for each $j \in D$ we have $\np_j
\le p_j - \delta$, and
  show that there must be a set of buyers $B(D)$ with the following properties:

\begin{enumerate}
\item Each buyer $i \in B(D)$ bought an item in $\nwin$ whose price dropped by at least $\delta$. Formally, letting $\nwin(j) = i$ we have $\np_{j} \le p_{j} -
\delta$.
\item The set is large: $|B(D)| \ge k+1$
\end{enumerate}

The existence of this set implies a contradiction. If $k+1$ buyers
bought items cheaper by $\delta$, how can only $k$ items are cheaper
by $\delta$?

To show that $B(D)$ exists, let
\[B(D) = \left(\{\win(f)\} \cup \{\supp(x) : x \in D\} \cup \{\win(x) : x \in D\}\right) \setminus \{\taker\}\]

We show that the first condition holds. Note that $\win(f)$ bought
an item whose price dropped $\delta$ by our assumption. Consider an
item $j \in D$ and a buyer $i$ with $\win(j) = i$ or $\supp(j) =
i$. If buyer $i$ did not get an item in $\nwin$,   then we have a
contradiction, since $v_{i,j|S} \ge \p_j \ge \np_j + \delta$. This
implies that each buyer in $\{\supp(x) : x \in D\} \cup \{\win(x) :
x \in D\}$ received an item which maximizes their utility given the
price vector $\np$.  Letting $\nwin(j') = i$, it must be that
\[v_{i,j'|S} - \np_{j'} \ge v_{i,j|S} - \np_{j}\]
but we also know that
\[v_{i,j|S} - \p_{j} \ge v_{i,j'|S} - \p_{j'}\;\;\;\;\mbox{and}\;\;\;\;\np_{j} \le \p_j - \delta,\]
which implies that $\np_{j'} \le \p_{j'} - \delta$.

We now show the second condition, that $|B(D)| \ge k+1$.
Since we are selling items using a support order, we have that
  \[|\{\supp(x) : x \in D\} \cup \{\winp(x) : x \in D\}| \ge k+1,\]
since the last item that is sold in $D$ has the supporter not
winning any item in $D$.

In addition, $\win(f) \not \in \{\win(x) : x \in D\}$ since all
the items in $D$ had strictly positive price in $\p$, and a buyer
cannot buy two items with positive price in an equilibrium. Finally,
$\win(f) \not \in \{\supp(x) : x \in D\}$ since the items are sold
in support order.
%due to the choice of the selling order.

Together this implies that $|\{\winp(f)\} \cup \{\supp(x) : x \in
D\} \cup \{\winp(x) : x \in D\}| \ge k+2$, and hence $|B(D)| \ge
k+1$.
\end{proof}

\begin{proof}{{\bf of Lemma \ref{no-underbid}}}
Consider a node $r$ of height $\ell$ where an item $f$ is sold, and
previously sold items are allocated using $S$. Assume that from the
$r$ the on equilibrium path is $(\winp,\pp)$.

If any player except $\winp(f)$ underbids, the payment and
allocation do not change, and the claim follows from the inductive
hypothesis.

If $\winp(f)$ underbids, he no longer gets the item and the buyer
$\supp(f)$ wins the item. The strategy tree continue to node $r'$, in which $\supp(f)$
wins the item. According to Lemma \ref{lem:f_prefer}, this is not desirable for $\win(f)$.
\end{proof}

\begin{lemma}
Assume that Claim~\ref{claim:SGP:ind} holds for any subtree of
height at most $\ell-1$. Then in any node of height $\ell$
overbidding is weakly dominated.
\end{lemma}

\begin{proof}
Consider a node $r$ of height $\ell$ where an item $f$ is sold, and
previously sold items are allocated using $S$. Assume that from the
$r$ the on equilibrium path is $(\winp,\pp)$.

If buyer $\winp(f)$ overbids, the allocation does not change and the
price goes up, so it is strictly dominated. A buyer $i \neq
\winp(f)$ overbids and does not get the item nothing changes (since
it is a first price auction we do not need to worry about payments).

The case we need to consider is that buyer $i$ overbids $\winp(f)$
and wins the item.
The strategy tree continue to node $r'$ with a new Walrasian
equilibrium define by $\winn$ and $\pn$. Let $Sn$ be the allocation
$S$ with the addition that $\supp(f)$ is allocated $f$.

%Similarly to the proof of Lemma \ref{no-underbid}, we apply the
%inductive hypothesis, to get that the rest of the allocation and the
%payments will be determined by the Walrasian equilibrium generated
%by giving the items in $S$ as they were given, and giving item $f$
%to buyer $i$. Denote this new Walrasian equilibrium by $\winn$ and
%$\pn$.

If buyer $i$ does not get an item in the new equilibrium, then the
deviation was not profitable, since by the inductive hypothesis
(Claim~\ref{claim:SGP:ind}) following deviations after $r'$ are not
profitable.
Let item $j$ be such that $\winn(j) = i$. If $\pn_j \ge \pp_j$, then
the deviation was not profitable since $\winp$ is an equilibrium (so
either $v_{i,j|S} \le \pp_j$ or buyer $i$ received an item he
prefers under $\pp$). So it must be that $\pn_j = \pp_j - \delta$
for some $\delta > 0$.

Again we let $D$ be a maximum size set such that for each $j \in D$
we have $\pn_j \le \pp_j - \delta$, and denote $|D| = k$. We show
that there must be a set of buyers $B(D)$ with the following
properties:
\begin{enumerate}
\item Each buyer $i' \in B(D)$ bought an item in $\winn$ whose price dropped by at least $\delta$. Formally, letting $\winn(j) = i'$ we have $\pn_{j} \le \pp_{j} -
\delta$.
\item The set is large: $|B(D)| \ge k+1$
\end{enumerate}
%And again this will imply a contradiction.
We define
  \[B(D) = \{i\} \cup \{\supp(x) : x \in D\} \cup \{\winp(x) : x \in D\}\]
Since we sell items in a support order, we have that $|B(D)|$ is at
least $k+1$, establishing the second condition. An in
Lemma~\ref{no-underbid}, each buyer in $B(D)$ has to win an item
whose price dropped.
\end{proof}

Ruling out underbidding and overbidding finishes the proof of Theorem \ref{thm:equilibrium}.

\subsection{Optimistic Conservative and  Uniqueness}

We would like to show that our strategy {\tt Unit-Wlrs-Eq} is the
unique pure optimistic conservative bidding strategy, given our game
tree. More precisely, we fix the tree order of selling the items and
the tie-breaking rules in each node. Given this, the only pure
sub-game perfect optimistic conservative bidding strategy is {\tt
Unit-Wlrs-Eq}.
We establish the following main result.

\begin{theorem}\label{uniq-unit}
The equilibrium {\tt Unit-Wlrs-Eq} is the unique pure optimistic
conservative sub-game perfect, up to degeneracies.
\end{theorem}

In this extended abstract we assume that valuations are generic, so there are no degeneracies.
Handling ties in the valuations requires a bit more care, but doesn't change the proof fundamentally.

\begin{proof}
The proof is based on induction from the leaves of the tree to the
root.

\noindent{\bf Base case}
The base case is the leaves of the tree. In each leaf $(S;f)$, we
have a set of items $S$ which was already allocated, and the last
item $f$ is being on sale.
%we have a first price
%auction on a single item $f$. Since $S$ includes all the items
%except $f$, essentially we are left with a single item first price
%auction.
Recall that $v_{i,f|S}$ is the residual valuation of buyer $i$ for
item $f$ given the allocation $S$.
Therefore we have a single item first price auction with valuations
$v_{\cdot,f|S}$. In this case, optimistic conservative bidding and
conservative bidding coincide. In particular, buyer $i$ can bid
above $v_{i,f|S}$, the residual value for buyer $i$ for item $f$
given the allocation of the items in $S$. The minimal revenue
Walrasian equilibrium has the highest valuation buyer winning the
item at the second highest valuation price.

Let $w=\winn(f)=\arg \max_i v_{i,f|S} $ be the buyer that maximizes
the residual valuation. According to the tie breaking rules, ties
are always broken in favor of $w$. Let $s=\supp(f)$ be the player
which maximizes $v_{i,f|S}$ where $i \neq \winn(f)$, i.e.,
$s=\arg\max_{i:i\neq w} v_{i,f|S}$. That is, $s$ has the second
highest valuation for the item, given the allocation $S$. Note that
since only $w$ can bid above $v_{s,f|S}$, and since ties are broken
in favor of $w$, no player would ever bid strictly above
$v_{s,f|S}$.

Now if the price is strictly below $v_{s,f|S}$, at most one of $s$
and $w$ would win the item, and the other would deviate and take it.
Therefore the price is $v_{s,f|S}$.
%{\tt [[YM The end of the old proof for the base case.]]}

\noindent{\bf Inductive step}
For the inductive step, an item $f$ is being sold, after a set $S$
was already sold. The inductive hypothesis is that any extension of
the allocation $S$ by allocating the item $f$ to one of the buyers,
the only optimistic conservative sub-game perfect equilibrium
is the minimal Walrasian equilibrium.
Let $(\winn, \pn)$ be a minimal Walrasian equilibrium with respect
to $v_{\cdot,\cdot |S}$.
Suppose now that $\winn(f)$ is not allocated item $f$ and item $f$
is allocated to a buyer $i$. In this case, due to the inductive
hypothesis, buyer $\winn(f)$ utility is determined by a minimal
Walrasian equilibrium, in which $S$ is extended to $S'$ by
allocating $f$ to $i$.

Lemma \ref{lem:f_prefer} shows
that this is not an equilibrium, since $\winn(f)$ would be better
off bidding $v_{\supp(f),f|S}$ and winning $f$ for sure.
%[[YM: This
%is not really true,since in the Lemma item $f$ was allocated to
%$\supp(f)$ and now we might have $i\neq \supp(f)$]]

Suppose that $\winn(f)$ is allocated the item $f$, but for some
price $p < \pn_f$. Due to the inductive hypothesis, we know that
after $f$ is allocated to $\winn(f)$, the rest of the minimal
Walrasian equilibrium plays out. The definition of supporting buyer
gives that $u_{\supp(f)} = v_{\supp(f),f|S} - \pn_f$, and therefore
buyer $\supp(f)$ is better off bidding $(p + \pn_f)/2$.
%For this reason if $\winn(f)$ is allocated the item $f$ it pays $\pn_f$.
%
Therefore, player $\winn(f)$ is allocated the item $f$ for the price $\pn_f$, and the theorem holds.
\end{proof}

\bibliography{firstprice}
\bibliographystyle{plain}
%\newpage
\appendix

\section{Uniqueness of the Optimistic Conservative Equilibrium of a Single Item First Price Auction}\label{app:first}

In this section we prove that the unique conservative optimistic equilibrium of a first price auction
of a single item is for the winner to buy it and pay the second price, even when
considering mixed equilibria.

Let  $f$ be an item sold in a first price auction with valuations
$v_{\cdot,f}$. In this case, optimistic conservative bidding and
conservative bidding coincide. In particular, buyer $i$ can not bid
above $v_{i,f}$. The minimal revenue
Walrasian equilibrium has the highest valuation buyer winning the
item at the second highest valuation price.

Consider some mixed equilibrium for this auction, where $F_i$ is the
cumulative function of buyer $i$. Conservative bidding means that
the bid of buyer $i$ is at most $v_{i,f}$, i.e., $F_{i}(v_{i,f})
= 1$. Let $w=\winn(f)=\arg \max_i v_{i,f} $ be the buyer that
maximizes the residual valuation. According to the tie breaking
rules, ties are always broken in favor of $w=\winn(f)$. Let
$s=\supp(f)$ be the player which maximizes $v_{i,f}$ where $i \neq
\winn(f)$, i.e., $s=\arg\max_{i:i\neq w} v_{i,f}$. That is,
$s=\supp(f)$ has the second highest valuation for the item, given
the allocation $S$. Note that since only $w=\winn(f)$ can bid above
$v_{s,f}$, and since ties are broken in favor of $w=\winn(f)$, no
player would ever bid strictly above $v_{s,f}$. I.e., for any
buyer $i$ we have $F_i(v_{s,f})=1$.

Let $b_i=\inf \{b:F_i(b)>0\}$. If $b_w < b_s$ then we claim that
buyer $w$ is not doing best response. Specifically, we can modify
the bid distribution of buyer $w$ to shift the distribution mass
from $[b_w,b_s]$, which is strictly positive, to $v_{s,f}$. Since
we have non-generic valuations, this will strictly increase the
utility. Therefore $b_w\geq b_s$.

We now consider whether there are mass points at $b_w$ and $b_s$.
Let $h_i=\sup_b \{b: F_i(b)<1\}$.
 Let us enumerate the four cases:
\begin{itemize}
\item Case $F_w(b_w)=0,F_s(b_s)=0$:
%Is this the case you want to have left?
%We consider $h_s$. [[YM: How do you conclude this?!]]
In this case $w$ is not best responding. Note that there is an
$\epsilon_0>0$, such that for any $\epsilon<\epsilon_0$ there is a
bid $b^\epsilon_w>b_w$ which has $F_w(b^\epsilon_w)\geq \epsilon$
and $F_s(b^\epsilon_w)<1$. This implies that the utility of $w$ is
at most $\epsilon$ for arbitrary small $\epsilon$ and therefore it
is $0$. However (due to the non-generic valuations) by bidding
$v_{s,f}$ buyer $w$ is guarantee a positive utility. Contradiction
to the assumption that this was an equilibrium.

\item Case $F_w(b_w)>0,F_s(b_s)=0$: This case $w$ is not BR and would do
better to shift the distribution mass from $b_w$ to $v_{s,f}$.

\item Case $F_s(b_s)>0$: In this case $s$ has zero utility when it bids $b_s$,
hence it has overall zero utility. Since $s$ bids conservatively it
implies that it either bids $v_{s,f}$ or losses with any other
bid. This implies that $h_s=\sup_b \{b: F_s(b)<1\}$ has $h_s\leq
b_w$. Therefore we need that $b_w=v_{s,f}$ otherwise buyer $s$ can
bid $(b_w+v_{s,f})/2$ and secure a positive utility. Since $w$
bids at most $v_{s,f}$ we have established that $w$ always bids
$v_{s,f}$.
\end{itemize}

\section{The Effect of the Selling Order} \label{app:order}
Part of the contribution of this work is to raise the algorithmic question of choosing the selling order. In this appendix, we present an example where the selling order leads to a factor of $m$ in the revenue, of the worse equilibrium.
We use the same market to show an example in which
the worse revenue of a bad order is far from the Walrasian revenue.

\begin{theorem}\label{thm-order-matters}
There is a market with $m$ items and $m+1$ unit demand buyers and maximal social welfare of $m + o(1)$ where
for some ordering of items there exists an SPE whose revenue is $1$, and
for a different ordering of items every SPE has
revenue at least $m$.
\end{theorem}

And also
\begin{theorem}\label{thm-walras}
There is a market with $m$ items and $m+1$ unit demand buyers where
the minimal Walrasian revenue is $m$,
and there exists an order of items and an SPE whose revenue is $1$.
\end{theorem}

%{\bf Avinatan: we say that the revenue depends on the order, but do we know
%to prove that there is an order in which the revenue is at least log m?}
%{\bf comment this out for now: the
%revenue depends on the order of the items to be sold. We show that
%there exists an order and a SGP equilibrium where the revenue is
%$1$. We show that for any order there exists a SGP equilibrium such
%that the revenue is $O(\log m)$ for $m$ items. In contrast the
%Walresian equilibrium prices each item at $1$ with e revenue of $m$.}

\begin{proof}[of Theorems \ref{thm-order-matters} and \ref{thm-walras}]
%Let $[m]$ denote the set of items and $[n]$ the set of buyers with $n = m+1$.
We begin by describing the market, denoting the items as $1, \ldots, m$, and the buyers
as $\{b_i \ :\  0 \le i \le m \}$.
Each buyer $b_i$ will have a set of items which interest it, and will have the same value of $1 + \epsilon i$ for all the
items in the set (we later take $\epsilon = 1/m^3$ to guarantee the bound on the social welfare).
Buyer $b_0$ is interested in item $1$, and values it for $1$. Buyer $b_i$ for $1 \le i \le m-1$ is an OR between items $i$ and  $i+1$,
and values each of them $1 + i \cdot \epsilon$.
Buyer $b_m$ is only interested in item $m$, and values it $1 + m \cdot \epsilon$.

We start with a bad SPE equilibrium for the order $1, \ldots, m$.
\begin{claim}
There exists an order of items and an SPE whose revenue is $1$.
\end{claim}

\begin{proof}
We apply Theorem \ref{folks} with the allocation $\omega(i) = i$ and prices $p(1)=1$ and $p(i)= 0$,
for $i\geq 2$. That is, item 1 goes to buyer 1 for $1$, item $i$  goes to buyer $i$ for 0, and buyer $0$ gets no items.

Intuitively, buyers $0$ and $1$ compete for item $1$, and buyer $1$
wins. Now there is no more competition in the market, and all items
have price zero.
\end{proof}

Since in equilibrium at least one buyer won't get
any items, it easy to see that the revenue of any equilibrium is at least 1.

We now consider what would happen if we sell the items in the order
$m, m-1, \ldots 1$.

\begin{claim}
In any SPE, for any buyer $b_i$ with $i \ge 1$, we have $\omega(i) = b_i$.
\end{claim}
\begin{proof}
  The proof is by induction. Buyer $b_m$ values $m$ more than
  $b_{m-1}$ does, and $b_m$ only values item $m$. Therefore, in any equilibrium
  $\omega(m) = b_m$. Now if buyer $b_{i+1}$ gets item $i+1$ in step $m - i$ of the auction,
  then in the leftover market buyer $b_i$ is only interested in item $i$. Since $b_i$
  values item $i$ more than buyer $b_{i-1}$ does, buyer $b_i$ must get it, and $\omega(i) = b_i$.
  \end{proof}

We now turn to the prices

\begin{claim}
  In any SPE we have $p_1 \le p_2 \cdots \leq p_m$
\end{claim}
\begin{proof}
  Suppose that $p_{i} > p_{i+1}$ for some $i$. But since we know buyer $i$ gets item $i$, and item
  $i+1$ is sold before item $i$, that would make buyer $i$ regret not bidding more for item $i+1$.
\end{proof}

Finally, $p_1 \ge 1$, as buyer $b_0$ gets no items, and we have $v_{b_0}(\{1\}) = 1$. This gives a total revenue of at least $m$,
and proves Theorem \ref{thm-order-matters}.

To prove Theorem~\ref{thm-walras}
we show that the Walrasian revenue is $m$. It is easy to see that
it is at most $m$, by using the price vector $p_i = 1$ for every $i$, and $\omega(b_i) = i$.

To show that it is at least $m$, we give a more
general claim. Suppose that all the buyers are unit demand and given
a buyer $i$ and an item $j$ either $v_i(j) \ge 1$ or $v_i(j) = 0$. Let
$S_i = \{j \ : \ v_i(j) \ge 1\}$.
%Let $v_i(T)=1$ if $T\cap S_i\neq \emptyset$. [[YM: why do we need this?]]
Given a set of items $T$ let $B(T)$ be the set of buyers
that want items in $T$, i.e., $B(T)=\{i: S_i\cap T\neq \emptyset\}$.
A set of buyers is ``complete" if for every set of items $T$ we have
$|T| < |B(T)|$.

\begin{claim}
For a complete set of buyers the Walrasian prices for all items price
are at least $1$.
\end{claim}

\begin{proof}
Assume that item $j$ is sold at a price $p<1$. Then any buyer in
$B(j)$ buys an item for price at most $p$. Let $T_1$ be the set of
items bought by the buyers in $B(j)$. Now we can consider $B(T_1)$.
Again, we can continue the process. Since the system is ``complete''
we are guaranteed that we can add more buyers in each iteration. At
the end we will have the set of items include all items, and there
are still buyers that did not buy any item. Those buyers can bid
$p'\in(p,1)$ and have a positive utility. A contradiction that some
item was sold at a price $p<1$.
\end{proof}

Since the set of all items in the theorem is complete, the minimal
Walrasian revenue is $m$. This completes the proof of Theorem~\ref{thm-walras}.
\end{proof}

\section{Demand Reduction and More Complex Valuations}\label{app:complex}
The definition of optimistic conservative bidding allows us to to disqualify threats which are not credible (see \ref{app:con}), and choosing the order wisely lets us build up auction pressure. However, these two are not enough to guarantee (say) the optimal outcome for large families of valuations. One challenge that we see is demand reduction. Indeed, the simplest market beyond unit demand buyers and additive buyers is a market that contains both additive and unit demand buyers.

Consider such a market with two buyers {\tt OR} (unit demand) and {\tt SUM} (additive), and two identical items $1,2$.
Player {\tt OR} is an or player, and values each item for $4$ (but also values the items together for $4$). Player {\tt SUM}
is additive, and values each item for $5$, and both of them for $10$. Since the items are identical, selling order is not an issue.

The Walrasian equilibrium is for {\tt SUM} to buy both items, and pay $4$ on each.

To us the only SPE that makes sense is that both players bid $1$ on the first item, and {\tt OR} gets it.\footnote{If tie breakers are for {\tt SUM} then {\tt SUM} can bid a distribution which has supermum $1$ and is very close to it}. Then, the buyer {\tt SUM} gets the second item for free. In the off path, if somehow {\tt SUM} got the first item, {\tt SUM} also gets the second item, but pays $4$ for it.

In any (optimistic) conservative equilibrium, the price of the first item can not be strictly less than 1, since then {\tt SUM} would want to buy it, and get the second item for $4$.

Note that there are many possible equilibria, but any prediction which does not give rise to this one seems problematic to us. Also note that the equilibrium we presented is the only optimistic conservative equilibrium.

In contrast to the bad equilibrium with the additive players presented in Section \ref{sec:additive} (which was a problem in equilibrium selection), or the bad example presented in \cite{LemeST12-soda} (which was a problem in the selling order), here the problem is with the benchmark - we should not expect optimal welfare, nor any decent revenue\footnote{Note that as the value of {\tt OR} increases the revenue decreases. The reason for that is that as the value rises, {\tt OR} can place a more credible threat on the second item, which incentivizes {\tt SUM} not to take the first item}. The solution to this problem would need to consider new techniques and new benchmarks, which could be of wider interest for understanding demand reduction. 
\section{Optimistic Conservative Bidding}\label{app:con}

One of the greatest challenges for us was to define the correct generalization of conservative bidding to sequential first price auctions.
Recall the original motivation for conservative bidding. Consider a second price auction for a single item, with two buyers $A,B$, where $A$ values
the item for $10$, and $B$ values the item for $5$. One possible equilibrium is that $A$ bids 0, $B$ bids $100$, and $B$ wins paying nothing. This is an equilibrium, since $A$ doesn't want to outbid $B$. We can see this as a type of threat (which is not credible) which $B$ poses to $A$, allowing him to buy the item only if $A$ is willing to pay a large amount.

One could hope that in first price auctions there would be no need for this assumption, since this particular bad example can no longer happen. However, we see something similar in sequential first price auctions, where buyers can pose threats which are not {\em really} credible, but are still formally allowed by the definition of the SPE - take the funny equilibrium in Section \ref{sec:additive} for example.

It is known in economics that SPE is a problematic solution concepts, and one of the natural refinements is trembling hand perfect equilibrium. In a trembling hand equilibrium, each player needs to play best response even if there is a vanishing probability that other players will play some random move.

The problem is that using this notion in auctions as is doesn't work well with the fact that prices are continuous. Indeed take the example we had before, with $A$ valuing the item for $10$, and $B$ valuing the item for $5$, and suppose that the item is sold in first price auction, with ties broken in favor of $A$.

The ``right'' equilibrium, namely $A$ and $B$ both bid $5$ and $A$ wins is not a trembling hand equilibrium, since if $A$ lowers his bid in any way (or doesn't participate in the auction or any other definition of tremble), then bidding $5$ is never a best response for $B$. Indeed, bidding $5$ guaranteers that $B$ walks away with zero utility, no matter what $A$ does. Worse, since the equilibrium where $A$ pays 5 and wins is the only equilibrium under these conditions, the auction has no trembling hand perfect equilibrium.\footnote{One could remedy this by letting the players bid only integer numbers, and limit the possible ways in which $A$ can tremble based on the valuations. This creates its own problems.} Part of the problem is that $B$ has zero utility - but in any auction someone is just holding the price and has zero utility.

As an alternative, we wanted to choose a definition of conservativeness which would prevent buyers from posing threats which are composed of bids that they do not want to win (in the spirit of not allowing threats which are not credible), and also allow for buyers who have zero utility in the equilibrium (such as buyer $B$ in the previous example). Optimistic conservative bidding does just that - it treats every bid that a buyer make as a ``threat'' and requires that if the buyer actually (in some weird off path equilibrium) needs to implement the threat, the buyer is not worse off than in the on path.

For the leaves of the tree, this notion coincides with the usual notion of conservative bidding.

We note that it is enough for all our results that the second highest buyer (who holds the price) is conservative optimistic, but there is no reason to limit the notion only to the second highest bidder.

\section{Submodular valuations - Highest Revenue Equilibrium}\label{app-opt}

In this appendix we prove Corollary \ref{high-revenue-submodular}.
We consider the best equilibrium from the seller's
perspective, and analyze how much revenue it elicits, where our
benchmark here will be the {\em full revenue} which is defined as
the social welfare of the optimal allocation (denoted as $OPT$).

Consider a market with sub-modular buyers. Let
$\pi$ be the following selling order, defined inductively. Let $i_1,
j_1$ the buyer and item that maximize $\max_{i,j} v_i(j)$. Now at
step $k$, let $i_k, j_k$ be the pair of buyer and item who maximize
\[\vm_i(j|G) = v_i(\{x_i(G)\} \cup \{j\}) - v_i(\{x_i(G)\} )\]
where $G$ is the state of the game where buyer $i_r$ wins item $j_r$ for every
$r < k$.

We remind Corollary \ref{high-revenue-submodular}:

\noindent {\em Assume all buyers have a submodular valuation.
If the items are sold according to $\pi$, there is an SPE which
elicits half of the full revenue, or $OPT/2$.}
%{\bf [[YM: full
%revenue is another word for SW, right?! What is OPT ?! Is OPT also
%SW?! - added a definition to the blurb in the section]]}

\begin{proof}
Define the equilibrium path inductively. In the first vertex, let
$i_1$ be the buyer who maximizes $\vm_i(\pi(1))$. Player $i_1$ gets
$\pi(1)$ and pays $\vm_{i_1}(\pi(1))$. In step $j$, let $G_j$ denote
the vertex of the tree. Let $i_j$ be the buyer who maximizes
$\vm_i(\pi(j) | G_j)$. Player $i_j$ wins item $\pi(j)$ and pays
$\vm_{i_j}(\pi(j)  |  G_j)$.

Using Theorem \ref{folks}, one can complete this to an SPE, as for
any $i \neq i_j$ buying item $\pi(j)$ for strictly more than
$\vm_{i_j}(\pi(j)  |  G_j)$ would result in negative utility.

Finally, as the greedy algorithm gives a $2$ approximation to
maximizing a submodular function, this allocation achieves half the
social welfare. As the buyers end up with zero utility, the revenue
must be at least $OPT/2$ as required.
\end{proof}

\end{document}